\documentclass[preprint,12pt,authoryear]{elsarticle}
\usepackage[T1]{fontenc}
\usepackage[utf8]{inputenc}
\usepackage{amsmath,amsthm,amssymb}
\usepackage{newtxtext,newtxmath}
\usepackage[margin=1in]{geometry}
\usepackage{setspace}
\usepackage{microtype}
\usepackage{booktabs,array,tabularx,longtable}
\usepackage{graphicx}
\usepackage{algorithm}
\usepackage{algpseudocode}
\usepackage{enumitem}
\usepackage{placeins}
\usepackage{needspace}
\usepackage[hidelinks]{hyperref}
\usepackage{url}
\newtheorem{theorem}{Theorem}[section]
\newtheorem{lemma}[theorem]{Lemma}
\newtheorem{proposition}[theorem]{Proposition}
\newtheorem{corollary}[theorem]{Corollary}
\theoremstyle{definition}
\newtheorem{assumption}[theorem]{Assumption}
\newtheorem{remark}[theorem]{Remark}
\newcommand{\Om}{O_m}
\newcommand{\OmK}{O_{m,K}}

\newcommand{\Z}{\mathbb{Z}}

\newcommand{\dom}{\operatorname{dom}}

\newcommand{\LS}{\textnormal{LS-PC}}

\newcommand{\nvec}{\boldsymbol{n}}
\newcommand{\Nvec}{\boldsymbol{N}}
\newcommand{\evec}{\boldsymbol{e}}
\setlist{nosep}
\numberwithin{equation}{section}
\makeatletter
\def\ps@pprintTitle{\let\@oddhead\@empty\let\@evenhead\@empty\def\@oddfoot{\hfil\thepage\hfil}\let\@evenfoot\@oddfoot}
\makeatother
\newcommand{\GeoUncap}{2.531}
\newcommand{\GeoCap}{1.006}
\newcommand{\GeoAll}{1.595}
\newcommand{\NewWins}{186}
\newcommand{\OldWins}{114}
\journal{European Journal of Operational Research}
\hypersetup{pdftitle={Monge matrix searching for lot sizing with piecewise-concave production costs},pdfauthor={Kleitos Papadopoulos}}
\begin{document}
\begin{frontmatter}
\title{Monge matrix searching for lot sizing with piecewise-concave production costs}
\author{Kleitos Papadopoulos}
\ead{kleitospa@gmail.com}
\begin{abstract}
We study single-item lot sizing with piecewise-concave production and holding--backlog costs over $T$ periods. Production has $m$ common positive finite breakpoints, and stock costs and domains have $K$ common finite boundary levels. For fixed $m,K$, we give an exact deterministic algorithm using $\OmK(T^{m+2}\alpha(T+2))$ arithmetic operations and a Las Vegas algorithm using $\OmK(T^{m+2})$ expected operations, where $\alpha$ is the inverse Ackermann function. The algorithms retain nonlinear concave cost pieces and cover backlogging, common-grid stock limits, and explicit terminal-stock policies. A structural decomposition at stock boundary levels yields breakpoint-only prefix and suffix tables. Batching transitions by a designated production period produces implicit matrices that are Monge double staircases on each production piece. Established matrix searches, merged state orders, and backward evaluation provide the bounds. A simpler rectangle/SMAWK variant requires $\OmK(T^{m+2}\log(T+2))$ operations. On 300 no-backlogging benchmark cases, its objectives agree exactly with the predecessor dynamic program and an inventory-state oracle. Geometric-mean paired time ratios are \GeoUncap\ for uncapacitated cases and \GeoCap\ for capacitated cases, with ratios above one favoring matrix searching. A 640-case base-model validation suite and a separate 624-case generalized-model suite pass. The timing study does not benchmark the generalized models or the stronger staircase routines. All bounds assume exact arithmetic, constant-time cost queries, and explicit closed-domain concavity conditions.
\end{abstract}
\end{frontmatter}

\section{Introduction}\label{sec:intro}
Production and procurement decisions often depend on quantity thresholds. A threshold may separate production regimes or mark a minimum order quantity or capacity. Within a regime, total cost can be concave even though the complete production-cost function is not globally concave. Holding and backlog costs can likewise change regime at specified stock levels. We study single-item lot sizing with these piecewise-concave costs, denoted by $\LS$, over a finite horizon. Negative stock represents backlogged demand, and unavailable production or stock ranges are imposed through the cost domains.

\citet{koca2014} established polynomial solvability of the common-production-breakpoint problem when the number of breakpoints is fixed. Their main algorithm uses the structural result that an optimal zero-inventory regeneration interval contains at most one production quantity outside the breakpoint set, together with a breakpoint-only final interval when positive terminal inventory is necessary. The underlying structure originates in \citet{swoveland1975}. Their additional cases include backlogging and piecewise-concave holding costs, for which regeneration states must record stock boundary levels. The computational difficulty is to combine the breakpoint production schedules before and after the exceptional production period.

For $T$ periods and $m$ common positive finite production breakpoints, the general no-backlogging algorithm of \citet{koca2014} requires $O(T^{2m+3})$ time. Its common-capacity restriction has an $O(T^{2m+2})$ bound. For related piecewise-linear models, \citet{ou2017} obtains $O(T^{m+2}\log T)$ time, and \citet{papadopoulos2026} obtains $O(T^{m+2})$ deterministic time. These results motivate ordered states, but affine separability does not directly apply to a nonlinear concave production piece. The cited models are not identified with one another in their treatment of backlogging or terminal inventory.

Our contribution is a matrix-searching reduction that retains nonlinear concave production pieces and admits the stock-boundary states required by the generalized model. For a fixed designated production period $t$, each transition has the form
\begin{equation}\label{eq:intromatrix}
 M_t(i,j)=A_i+C_j+p_t(a_j-\beta_i).
\end{equation}
A row represents a breakpoint-only prefix with a prescribed starting stock boundary. A column represents a breakpoint-only suffix ending at a stock boundary, followed by an optimal continuation. The quantity $a_j-\beta_i$ is the production needed in period $t$. Ordering these inventory keys exposes a Monge inequality on each concave production piece. Its defined entries form a double staircase rather than an arbitrary complete matrix. Holding at period $t$ is a column offset, so piecewise-concave stock costs do not create a further row--column interaction.

Let $K$ be the number of common stock boundary levels, including zero and the initial stock. We combine the reduction with deterministic and randomized staircase-searching algorithms \citep{klawe1990,chan2021}. A linear-total-size decomposition, incrementally merged orders, and backward evaluation give $\OmK(T^{m+2}\alpha(T+2))$ deterministic time and $\OmK(T^{m+2})$ expected Las Vegas time, with $\OmK(T^{m+2})$ worst-case space. The boundary labels enlarge the states but do not add a power of $T$ for fixed $m,K$. These are arithmetic and cost-oracle bounds, not bit-complexity claims or claims of fixed-parameter tractability.

The no-backlogging problem with globally concave holding costs is the specialization $\mathcal H=\{0\}$, with only nonnegative stock allowed. Ordinary backlogging also has $K=1$ when stock costs are concave separately on the two half-lines. Additional common holding-cost breakpoints, inventory limits, and backlog limits are included in $\mathcal H$. The boundary-level structural principle is not new \citep{swoveland1975,koca2014}; the contribution is its integration with the faster nonlinear Monge batching. Production capacities, minimum quantities, and multiple production regimes are represented by the same production domains. Terminal policies are specified explicitly rather than treated as interchangeable.

We implement a simpler $\OmK(T^{m+2}\log(T+2))$ variant using complete-matrix searching by SMAWK \citep{aggarwal1987}. A 300-case no-backlogging experiment compares its original-model specialization with a new implementation of the original Koca--Yaman--Akt\"urk interval recurrences, not with naive enumeration of the new matrices. It records complete objective-only solve times. A separate generalized solver is validated on 624 policy-specific cases with signed stock, piecewise-concave stock costs, and other domain restrictions. The original 640-case validation suite is also retained. Only the original model has timing comparisons; the stronger staircase routines are not implemented or benchmarked.

Section~\ref{sec:model} defines the unified model and its optimal-solution structure. Section~\ref{sec:algorithm} derives the boundary-indexed tables, Monge searches, full algorithm, and correctness and complexity proofs. Section~\ref{sec:experiments} distinguishes the original timing study from the generalized-model validation. Section~\ref{sec:special} gives the principal specializations and scope restrictions, and Section~\ref{sec:conclusion} concludes.

\section{Problem definition and properties of optimal solutions}\label{sec:model}
\subsection{Production, signed stock, and cost domains}\label{subsec:model}
There are $T\geq1$ periods. Period $t$ has nonnegative demand $d_t$, production $x_t\geq0$, and end-of-period stock $s_t$. Positive stock is inventory and negative stock is backlogged demand. The function $p_t(x)$ is the \emph{total production cost}, not a unit price; $h_t(s)$ is the holding or backlog cost. Given initial stock $s_0$ and a terminal set $E$, the problem is
\begin{align}
 \min\quad &\sum_{t=1}^{T}\bigl(p_t(x_t)+h_t(s_t)\bigr),\label{eq:obj}\\
 \text{subject to}\quad &s_{t-1}+x_t=d_t+s_t &&(t=1,\ldots,T),\label{eq:balance}\\
 &x_t\in\dom p_t,\quad s_t\in\dom h_t &&(t=1,\ldots,T),\qquad s_T\in E.\label{eq:domains}
\end{align}
Production and stock feasibility are specified by the following domains. In particular, stock is not constrained to be nonnegative unless its domains impose that restriction.

\begin{assumption}[Piecewise-concave production costs]\label{ass:cost}
Fix $m\geq1$ and common breakpoints
\[
 0=b_0<b_1<\cdots<b_m<b_{m+1}=+\infty.
\]
For each $t$, a supplied family $\mathcal J_t$ consists of selected closed intervals $[b_{\ell-1},b_\ell]$ for $1\leq\ell\leq m$, possibly the ray $[b_m,+\infty)$, and possibly singletons $\{b_\ell\}$ for $0\leq\ell\leq m$. The actual function $p_t$ is finite, nonnegative, and concave on each supplied domain, and is $+\infty$ outside their union. Its actual values at shared endpoints agree across domains.
\end{assumption}
The effective production domain is $\dom p_t=\bigcup_{J\in\mathcal J_t}J\subseteq[0,+\infty)$. Missing intervals model unavailable production ranges; capacities on the common grid and isolated allowable quantities are included. After duplicate domains are removed, $q:=\max_t|\mathcal J_t|\leq2m+2$. The production functions and available domains may vary with time; their finite pieces need not be affine.

\begin{assumption}[Piecewise-concave stock costs]\label{ass:stock}
Let $\mathcal H=\{r_1<\cdots<r_K\}$ be a finite set containing zero and the prescribed initial stock $s_0$. In each period, the available stock domains are selected closed intervals $[r_j,r_{j+1}]$, the possible rays $(-\infty,r_1]$ and $[r_K,+\infty)$, and possible singletons $\{r_j\}$. The actual function $h_t$ is finite, nonnegative, and concave on each supplied domain, and is $+\infty$ outside their union. Values at shared endpoints agree. The set $\mathcal H$ and its size $K$ are common over time, although costs and domain availability may depend on the period.
\end{assumption}
A larger available concave interval can be split at extra levels of $\mathcal H$. Thus zero and $s_0$ can be added without changing any cost. Finite inventory or backlog limits must be included in $\mathcal H$ and enforced through the available domains.

We distinguish three terminal sets in~\eqref{eq:domains}:
\begin{equation}\label{eq:terminalsets}
 E=\{0\},\qquad E=[0,+\infty),\qquad E=\mathbb R.
\end{equation}
The first clears both stock and backlog; the second clears backlog but permits surplus; the third intentionally permits a final backlog with its specified cost. These are different optimization problems. Intersect the available domains of $h_T$ with $E$, assigning $+\infty$ outside $E$. Since $0\in\mathcal H$, this operation preserves Assumption~\ref{ass:stock}. Below, $h_T$ denotes this effective function. The predecessor's zero initial inventory is recovered by taking $s_0=0$.

For the original no-backlogging model, take $s_0=0$, $\mathcal H=\{0\}$, and a finite nonnegative concave holding cost on $[0,+\infty)$, with no negative stock domain. The choices $E=\{0\}$ and $E=[0,+\infty)$ recover its \emph{zero-terminal} and \emph{free-terminal} variants. The latter is the direct comparison with the main formulation of \citet{koca2014}. In the signed-stock model, ``free terminal'' is not used to obscure whether final backlog is allowed.

Both cost assumptions concern actual values on closed domains, not just interior formulas. Downward endpoint jumps, including a production setup discontinuity at zero, are permitted when each closed restriction remains concave. Arbitrary half-open domains or upward endpoint jumps are not covered merely because their interior formulas are concave. Larger concave domains can be split at added common breakpoints without changing costs. No global concavity of $p_t$ or $h_t$ is assumed.

\subsection{Notation and computational model}\label{subsec:notation}
Let $Q_0=0$, $Q_t=\sum_{k=1}^{t}d_k$, and $D(u,v)=Q_v-Q_{u-1}$, with $D(u,u-1)=0$. For $k\geq0$ define
\begin{equation}\label{eq:counts}
 \mathcal N_k=\{\nvec\in\Z_{\geq0}^{m}:|\nvec|\leq k\},\qquad
 |\nvec|=\sum_{\ell=1}^{m}n_\ell,\qquad
 w(\nvec)=\sum_{\ell=1}^{m}n_\ell b_\ell.
\end{equation}
The count $n_\ell$ records productions at $b_\ell$. The remaining $k-|\nvec|$ periods produce zero, so zero-production periods are counted implicitly, not omitted. Consequently $|\mathcal N_k|=\binom{k+m}{m}$. Write $\evec_\ell$ for the $\ell$th unit vector, with $\evec_0=0$. Invalid states and minima over empty sets have value $+\infty$.

One exact arithmetic operation, comparison, array access, or supplied cost evaluation counts as one operation. Randomized searching additionally uses unit-cost independent random words and uniform index sampling. The notation $\OmK(\cdot)$ permits constants depending on fixed $m,K$; $\Om(\cdot)$ is used when $K=1$. These notations do not assert fixed-parameter tractability. The function $\alpha$ is the inverse Ackermann function. No constant-time manipulation of arbitrary symbolic representations or bounded bit lengths for exact numbers is assumed.

\subsection{Existence of an optimum}\label{subsec:attainment}
\begin{lemma}[Attainment]\label{lem:attain}
Every feasible suffix starting at a level $\gamma\in\mathcal H$ has a finite attained optimum under each terminal convention in~\eqref{eq:terminalsets}.
\end{lemma}
\begin{proof}
Let $D_*=D(1,T)$ and $B=\max\{b_m,D_*+r_K-r_1\}$. Consider a feasible suffix starting with $\gamma\in\mathcal H$. If some production $x_i>B$, replace it by $B$. For every affected period $k\geq i$, the new stock is at least
\[
 \gamma+B-D(u,k)\geq r_1+B-D_*\geq r_K,
\]
where $u$ is the suffix's starting period. The original stock was strictly larger. Hence each affected finite stock cost is on an available rightmost ray, and the new stock remains on that ray. The changed production also remains on its available rightmost ray. A nonnegative concave function on a right ray is nondecreasing, so neither production nor holding cost increases. The effective terminal restriction is preserved; if its right ray is unavailable, a production exceeding $B$ could not have been feasible in the first place.

Repeating bounds every production by $B$. The resulting nonempty feasible set is compact: production is bounded and stock is determined by balance. Each closed concave restriction is lower semicontinuous; the actual shared endpoint values are consistent, and the domain families are finite. The extended-value objective is therefore lower semicontinuous on that set and attains its minimum.
\end{proof}

Nonnegativity gives a sufficient well-posedness condition. The structural and algorithmic arguments also apply to real-valued finite costs under the same domain and concavity assumptions whenever every feasible suffix starting at a level in $\mathcal H$ has a finite attained optimum.

\subsection{Boundary intervals and exceptional production}\label{subsec:structure}
The structural principle below is the boundary-level counterpart of the regeneration property used by \citet{koca2014}; see also \citet{swoveland1975} and Remark~1 and Section~3.4 of \citet{koca2014}. We include its proof to make the endpoint, stock-domain, and terminal assumptions explicit.

Call $t$ a \emph{boundary time} when $s_t\in\mathcal H$. A boundary interval $[u,v]$ has $s_{u-1},s_v\in\mathcal H$ and no internal boundary time. A final tail has no boundary time after its starting boundary, including at $T$. A fractional production is a quantity outside $\{b_0,\ldots,b_m\}$, not necessarily a noninteger quantity.

\begin{lemma}[Boundary-interval structure]\label{lem:structure}
There is an optimum in which each boundary interval has at most one fractional production. If its terminal stock does not belong to $\mathcal H$, the final tail uses breakpoint productions only.
\end{lemma}
\begin{proof}
Among all optima, choose one maximizing the number of boundary times and, subject to that, minimizing the number of fractional productions. Existence follows from Lemma~\ref{lem:attain} and the finiteness of these integer criteria.

Suppose fractional periods $i<j$ occur in one boundary interval. Let $[\ell_i,u_i]$ and $[\ell_j,u_j]$ be their available production pieces. Each affected stock $s_k$, $i\leq k<j$, lies strictly inside an available stock piece $[a_k,c_k]$; an endpoint can be infinite. Change $x_i$ to $x_i+z$ and $x_j$ to $x_j-z$. The affected stocks change by $z$. This perturbation stays in those same pieces for
\[
 -\delta_-\leq z\leq\delta_+,
\]
where
\begin{align*}
 \delta_-&=\min\{x_i-\ell_i,\ u_j-x_j,\ \min_{i\leq k<j}(s_k-a_k)\},\\
 \delta_+&=\min\{u_i-x_i,\ x_j-\ell_j,\ \min_{i\leq k<j}(c_k-s_k)\}.
\end{align*}
Both numbers are positive and finite, since the first contains $x_i-\ell_i$ and the second contains $x_j-\ell_j$. The objective along the segment is concave. Optimality at the interior point $z=0$ forces both endpoint costs to equal its cost. At either endpoint, a production reaches a finite production breakpoint or an affected stock reaches a finite level in $\mathcal H$. No previous boundary time changes. A new boundary contradicts maximality; otherwise a fractional production is eliminated, contradicting the second criterion.

For a fractional period $i$ in a final tail, change only $x_i$ by $z$ and all stocks from $i$ through $T$ by $z$. A negative step of length
\[
 \delta=\min\{x_i-\ell_i,\ \min_{i\leq k\leq T}(s_k-a_k)\}>0
\]
reaches a production or stock boundary. A sufficiently small positive step is also feasible in the same pieces, including the effective terminal domain. Concavity and optimality again make the negative endpoint optimal, contradicting one of the two selection criteria.
\end{proof}

For ordinary backlogging, $\mathcal H=\{0\}$ and $h_t$ is concave on each of $(-\infty,0]$ and $[0,+\infty)$. Internal stock is required to be nonzero, not positive. Its sign need not remain the same throughout an interval: a balance step may move from negative to positive stock without equalling zero. The proof uses a separate concave stock piece at each affected time and therefore covers such a sign change.

\begin{corollary}[Lattice optimum]\label{cor:integer}
If demands, production breakpoints, all levels of $\mathcal H$, and $s_0$ belong to $\Delta\mathbb Z$ for some $\Delta>0$, every feasible instance has an optimum with all production and stock quantities in $\Delta\mathbb Z$.
\end{corollary}
\begin{proof}
In a boundary interval from $\gamma$ to $\eta$, total production is $D(u,v)+\eta-\gamma$. All but at most one production are breakpoint quantities, so the remaining production is on the same lattice. A possible final tail consists entirely of breakpoint productions. Balance gives the stock quantities.
\end{proof}

\paragraph{Why the extra boundary states are necessary.}
Take $T=2$, demands $(0,4)$, admissible production $[0,6]$, $p_1(x)=p_2(x)=x(6-x)$, $h_1(s)=20|s-2|$ for $s\geq0$, $h_2(s)=0$, and $s_0=s_2=0$. The production breakpoint set is $\{0,6\}$ and the stock boundary set is $\{0,2\}$. The optimum is $(x_1,x_2)=(2,2)$, with cost $16$ and $s_1=2$. There are two fractional productions between the two zero-stock times, but only one in each interval separated at stock $2$. Restricting this zero-to-zero block to one fractional production leaves only $(0,4)$ and $(4,0)$, both of cost $48$. For $x_1\in[0,4]$, the total cost is $8+8x_1-2x_1^2+20|x_1-2|$, which is concave on each side of $2$; comparing $x_1\in\{0,2,4\}$ proves the stated optimum. Thus merely evaluating a more general holding-cost oracle in a zero-boundary-only algorithm is not a valid extension. This example is also included in the executable validation suite.

The algorithm may allow internal visits to $\mathcal H$ in a candidate block. Such visits add feasible candidates without removing the structured optimum. A designated period need not actually be fractional: this convention represents breakpoint-only boundary intervals as well.

\section{An improved exact algorithm}\label{sec:algorithm}
We develop one algorithm for the model of Section~\ref{sec:model}. It separates breakpoint-only schedules on either side of a designated period and records the starting and ending stock boundary levels. The original no-backlogging algorithm is obtained by keeping only the boundary level zero and the nonnegative stock domain.

\subsection{Boundary-indexed prefix and suffix tables}\label{subsec:tables}
For $\gamma\in\mathcal H$, let $L^\gamma_{u,k}(\nvec)$ be the minimum cost on periods $u,\ldots,k$ with incoming stock $s_{u-1}=\gamma$ and breakpoint-only production counts $\nvec$. Holding through $k$ is included, and its final stock is
\[
 s_k=\gamma+w(\nvec)-D(u,k).
\]
Set $L^\gamma_{u,u-1}(\boldsymbol0)=0$, and give all other empty-block counts value $+\infty$. For $\nvec\in\mathcal N_{k-u+1}$,
\begin{align}
 L^\gamma_{u,k}(\nvec)
 &=h_k\bigl(\gamma+w(\nvec)-D(u,k)\bigr)\notag\\
 &\quad+\min_{0\leq\ell\leq m}
 \{L^\gamma_{u,k-1}(\nvec-\evec_\ell)+p_k(b_\ell)\}.\label{eq:left}
\end{align}
An invalid predecessor or infinite cost excludes that transition. In particular, the zero-production predecessor must belong to its own, shorter count set.

For $\eta\in\mathcal H$, let $R^\eta_{k,v}(\nvec)$ be the minimum cost on periods $k,\ldots,v$ with breakpoint-only counts $\nvec$ and $s_v=\eta$. Its required incoming stock is
\[
 s_{k-1}=\eta+D(k,v)-w(\nvec).
\]
Holding at $k-1$ is not included. Set $R^\eta_{v+1,v}(\boldsymbol0)=0$, with all other empty-block counts infinite. For $\nvec\in\mathcal N_{v-k+1}$,
\begin{align}
 R^\eta_{k,v}(\nvec)
 &=\min_{0\leq\ell\leq m}\bigl\{p_k(b_\ell)
 +h_k\bigl(\eta+D(k+1,v)-w(\nvec-\evec_\ell)\bigr)\notag\\
 &\hspace{39mm}+R^\eta_{k+1,v}(\nvec-\evec_\ell)\bigr\}.\label{eq:right}
\end{align}
Only valid predecessors and finite costs are admitted. There is no blanket nonnegative-stock test: the supplied $h_k$ enforces its stock domains, and the effective $h_T$ enforces the terminal restriction. Each table checks end-of-period stock within its own interval. The right table's incoming stock is checked when the preceding period is attached; no cost at $k-1$ belongs in $R^\eta_{k,v}$. A left table with prescribed incoming $\gamma$ likewise excludes the preceding period's cost.

\begin{lemma}[Table correctness and cost]\label{lem:tables}
Equations~\eqref{eq:left} and~\eqref{eq:right} compute their stated optimal values. For fixed $m,K$, all tables and minimizing production choices require $\OmK(T^{m+2})$ time and space.
\end{lemma}
\begin{proof}
Conditioning on the final production of a left block gives~\eqref{eq:left}; its ending stock is fixed by the count and incoming boundary. Conditioning on the first production of a right block gives~\eqref{eq:right}; its incoming and next stock follow from its ending boundary and counts. Every admitted predecessor extends to a feasible schedule for the relevant interval, and every feasible schedule has one such predecessor. Induction on interval length proves optimality, including all charged stock-domain tests.

There are $O(T^2)$ intervals, $K$ boundary labels per table, and at most $\binom{T+m}{m}=\Om(T^m)$ counts per interval and label. A state examines at most $m+1$ choices. Weights take $O(m)$ arithmetic operations and can be stored. Arrays indexed by interval endpoints, boundary indices, and count coordinates give deterministic constant-time access for fixed $m,K$. Witnesses preserve these bounds.
\end{proof}
Breakpoint-only count subproblems are part of the established dynamic-programming approach, not a separate novelty claim \citep{koca2014,ou2017,papadopoulos2026}. Define the breakpoint-only terminal alternative
\begin{equation}\label{eq:tail}
 Z_{u,\gamma}=\min_{\nvec\in\mathcal N_{T-u+1}}L^\gamma_{u,T}(\nvec).
\end{equation}
This definition already enforces the terminal policy through $h_T$. It is valid for all three terminal sets. Under zero terminal stock it supplies redundant but feasible breakpoint-only alternatives; these are also represented by the main transitions.

\subsection{An exact batched Bellman equation}\label{subsec:batch}
Let $\Psi_{u,\gamma}$ be the true optimum on periods $u,\ldots,T$ with incoming stock $\gamma\in\mathcal H$, excluding any cost at $u-1$. Initialize $\Psi_{T+1,\eta}=0$ for $\eta\in\mathcal H\cap E$ and $+\infty$ otherwise. Section~\ref{subsec:sweep} gives an evaluation order without circular dependencies.

Fix a designated period $t$. An admitted row $i=(u,\gamma,\nvec)$ has $1\leq u\leq t$, $\gamma\in\mathcal H$, $\nvec\in\mathcal N_{t-u}$, and finite prefix cost. Define
\begin{equation}\label{eq:row}
 \beta_i=\gamma+w(\nvec)-D(u,t),\qquad
 A_i=L^\gamma_{u,t-1}(\nvec).
\end{equation}
An admitted column $j=(v,\eta,\Nvec)$ has $t\leq v\leq T$, $\eta\in\mathcal H$, $\Nvec\in\mathcal N_{v-t}$, and finite column cost, where
\begin{align}
 a_j&=\eta+D(t+1,v)-w(\Nvec),\notag\\
 C_j&=R^\eta_{t+1,v}(\Nvec)+h_t(a_j)+\Psi_{v+1,\eta}.\label{eq:col}
\end{align}
Both inventory keys may be negative. The actual inventory before production in $t$ is $\beta_i+d_t$, while $a_j$ is the required stock just after $t$. Their bridge production is
\begin{equation}\label{eq:bridge}
 x_t=a_j-\beta_i
 =D(u,v)+\eta-\gamma-w(\nvec)-w(\Nvec).
\end{equation}
Consequently the implicit transition cost is
\begin{equation}\label{eq:matrix}
 M_t(i,j)=A_i+C_j+p_t(a_j-\beta_i),
 \qquad a_j-\beta_i\in\dom p_t.
\end{equation}
All admitted entries are finite; entries outside the production domain are undefined.

\begin{lemma}[Exact transition identity]\label{lem:transition}
Equation~\eqref{eq:matrix} gives the minimum cost for its prescribed prefix, designated production, boundary-ending suffix, and optimal continuation. Every defined entry is a feasible suffix candidate with its row's incoming boundary. Every boundary interval with at most one fractional production is represented.
\end{lemma}
\begin{proof}
Equation~\eqref{eq:bridge} is balance at period $t$. The left table ends with stock $\beta_i+d_t$ and the right table requires incoming stock $a_j$. The two tables enforce all other interval constraints; $h_t(a_j)$ enforces stock feasibility immediately after the designated production. The ending level $\eta$ agrees with the continuation's incoming level. Each production and stock cost is charged exactly once. In particular, $h_t(a_j)$ is a column cost, not a nonlinear row--column term.

For a structured boundary interval, designate its fractional period, or any period if all productions are breakpoints. Its endpoint levels and counts on either side give the required entry. The balance-determined designated quantity is evaluated by the original nonlinear production oracle; no interpolation or linear approximation is used.
\end{proof}
Define $\Phi_{u,\gamma,t}$ as the minimum of~\eqref{eq:matrix} over all admitted rows labelled $(u,\gamma,\nvec)$ and all defined columns.
\begin{proposition}[Batched Bellman equation]\label{prop:bellman}
For each terminal convention,
\begin{equation}\label{eq:bellman}
 \Psi_{u,\gamma}=\min\left\{Z_{u,\gamma},\ \min_{u\leq t\leq T}\Phi_{u,\gamma,t}\right\}.
\end{equation}
\end{proposition}
\begin{proof}
Every finite alternative is feasible, so it is no smaller than $\Psi_{u,\gamma}$. Conversely, select a structured suffix optimum from Lemma~\ref{lem:structure}. If it consists entirely of a breakpoint-only final tail, it is represented by $Z_{u,\gamma}$. Otherwise its first boundary interval is represented by Lemma~\ref{lem:transition}; replacing the remaining suffix by its optimum cannot increase the cost. This proves the reverse inequality. The terminal alternative checks the effective $h_T$, including under zero terminal stock.
\end{proof}

\subsection{Concavity and the partial Monge structure}\label{subsec:monge}
Order rows by nondecreasing $\beta_i$ and columns by nonincreasing $a_j$, breaking ties consistently by labels, including their stock boundary levels. For $J\in\mathcal J_t$, retain exactly the entries with $a_j-\beta_i\in J$, obtaining $M_t^J$.

\begin{lemma}[Piecewise Monge property]\label{lem:monge}
If $i<i'$ and $j<j'$ and all four entries are defined in $M_t^J$, then
\begin{equation}\label{eq:monge}
 M_t^J(i,j)+M_t^J(i',j')\leq M_t^J(i,j')+M_t^J(i',j).
\end{equation}
\end{lemma}
\begin{proof}
Let $x=a_{j'}-\beta_{i'}$, $\delta=\beta_{i'}-\beta_i\geq0$, and $\varepsilon=a_j-a_{j'}\geq0$. The four quantities are $x$, $x+\delta$, $x+\varepsilon$, and $x+\delta+\varepsilon$, all in $J$. Concavity implies
\[
 p_t(x)+p_t(x+\delta+\varepsilon)
 \leq p_t(x+\delta)+p_t(x+\varepsilon).
\]
The row and column offsets cancel, yielding~\eqref{eq:monge}. Zero increments and boundary quantities obey the same inequality.
\end{proof}
Nonlinear holding costs do not create a further row--column interaction: holding outside $t$ is already optimized in the tables, and $h_t(a_j)$ depends only on the column. This remains true when holding--backlog costs are only piecewise concave. Their concavity was needed for the structural decomposition, not for an additional matrix inequality. All starting and ending boundary labels are included in the same ordered batch; no separate Cartesian enumeration of their pairs is needed.

For a domain $J=[\ell,h]$, with $h$ possibly infinite, defined entries in row $i$ satisfy
\begin{equation}\label{eq:window}
 \beta_i+\ell\leq a_j\leq\beta_i+h.
\end{equation}
The columns form an interval $[\lambda_i,\rho_i]$, possibly empty. Both endpoints are nonincreasing down the ordered rows. Singleton domains satisfy the same statement. Reverse-order scans over the rows compute the two boundary sequences in $O(r+c)$ time for an $r\times c$ batch. Empty rows have minimum $+\infty$ and are omitted from the search. Deleting infeasible rows or columns preserves the order and interval property.

\begin{remark}[Undefined entries cannot simply be padded]\label{rem:padding}
For $\beta=(0,2)$, $a=(2,0)$, and $J=\{0\}$, the defined entries form the off-diagonal pattern
\[
 \begin{pmatrix}\text{undefined}&0\\0&\text{undefined}\end{pmatrix}.
\]
Filling the undefined entries with $+\infty$ violates~\eqref{eq:monge} and reverses the row-minimizer order. Consequently, a complete-matrix search on this padded matrix is invalid. We search the partial domain or completely defined submatrices.
\end{remark}

\subsection{Searching a double staircase}\label{subsec:stairs}
We use standard matrix-searching results, not new implementations of these primitives. A \emph{rising staircase} has a defined prefix in each nonempty row, with prefix lengths nonincreasing down the rows. With constant-time entry access, a complete $r\times c$ Monge matrix admits all-row-minimum search in $O(r+c)$ deterministic time by SMAWK \citep{aggarwal1987}. A Monge rising staircase admits $O((r+c)\alpha(r+c+2))$ deterministic search \citep{klawe1990}, or $O(r+c)$ expected Las Vegas search \citep{chan2021}. Minimizing column labels are retained. These rectangular bounds also follow from the square bounds by repeating rows or columns to size $\max\{r,c\}$, preserving the staircase and Monge properties. Consistent column tie-breaking can be viewed as an infinitesimal column perturbation.

The double-staircase decomposition is related to established generalized matrix-searching techniques \citep{aggarwal1990,chan2021}. The following explicit construction supplies the size accounting needed here.
\begin{lemma}[Linear-total-size decomposition]\label{lem:decompose}
Suppose the nonempty row intervals of a partial $r\times c$ Monge matrix are $[\lambda_i,\rho_i]$, with both endpoint sequences nonincreasing. Its entries can be partitioned into rising-staircase Monge matrices with at most $2r$ total row occurrences and $2c$ total column occurrences. The construction takes $O(r+c)$ time.
\end{lemma}
\begin{proof}
Take the first unprocessed row $s$, set $p=\lambda_s$, and let $e$ be the last consecutive row with $\rho_e\geq p$. Every row $s,\ldots,e$ contains $p$. Split this group into the left column span $[\lambda_e,p]$ and the right span $[p+1,\rho_s]$. In the right span, defined entries form prefixes of nonincreasing length. In the left span they form suffixes; reversing both row and column orders turns them into prefixes of nonincreasing length. Reversing both orders preserves the Monge inequality. Delete empty rows in the right part and continue at $e+1$.

Each original row occurs in at most two parts. Let $I_g=[\lambda_{e_g},\rho_{s_g}]$ be group $g$'s full column span. The next group satisfies $\rho_{s_{g+1}}<\lambda_{s_g}$. Applying this twice gives $\rho_{s_{g+2}}<\lambda_{s_{g+1}}\leq\lambda_{e_g}$. Nonadjacent spans are disjoint; each column belongs to at most two spans, and a group's parts have disjoint columns. Thus total row and column occurrences are at most $2r$ and $2c$. Scanning the boundaries and constructing the index maps costs their total size.
\end{proof}

\begin{algorithm}[t]
\caption{Decompose a double-staircase domain}\label{alg:decompose}
\begin{singlespace}
\begin{algorithmic}[1]
\Require Nonempty intervals $[\lambda_i,\rho_i]$, $i=1,\ldots,r$, with nonincreasing boundaries
\State $s\gets1$
\While{$s\leq r$}
 \State $p\gets\lambda_s$; $e\gets s$
 \While{$e<r$ and $\rho_{e+1}\geq p$} \State $e\gets e+1$ \EndWhile
 \State Emit rows $s,\ldots,e$ and columns $\lambda_e,\ldots,p$, reversing both orders
 \If{$p<\rho_s$}
  \State Emit rows $s,\ldots,e$ and columns $p+1,\ldots,\rho_s$, deleting empty rows
 \EndIf
 \State $s\gets e+1$
\EndWhile
\end{algorithmic}
\end{singlespace}
\end{algorithm}

\Needspace{6\baselineskip}
\begin{proposition}[Batch-search cost]\label{prop:search}
Given the ordered keys of an $r_t\times c_t$ batch, all row minima and witnesses take
\[
 O\bigl(q(r_t+c_t)\alpha(r_t+c_t+2)\bigr)
\]
deterministic time, or $O(q(r_t+c_t))$ expected Las Vegas time.
\end{proposition}
\begin{proof}
For each supplied domain, scan the boundaries, delete empty rows, apply Lemma~\ref{lem:decompose}, and search the resulting staircases. Their total dimensions are $O(r_t+c_t)$. Minimize across all domains for each original row. Every feasible quantity is included, and a shared endpoint always has the same actual cost, so its possible duplication is harmless.
\end{proof}

\subsection{Reusing sorted state orders}\label{subsec:orders}
Sorting every full batch independently would introduce an unnecessary logarithmic factor. We reuse orders on combinatorial state sets rather than on feasibility-filtered sets, which need not be nested. Define
\begin{align*}
 \mathcal R_t&=\{(u,\gamma,\nvec):1\leq u\leq t,\ \gamma\in\mathcal H,\ \nvec\in\mathcal N_{t-u}\},\\
 \mathcal C_t&=\{(v,\eta,\Nvec):t\leq v\leq T,\ \eta\in\mathcal H,\ \Nvec\in\mathcal N_{v-t}\}.
\end{align*}
The hockey-stick identity gives
\begin{equation}\label{eq:batchsizes}
 \bar r_t=|\mathcal R_t|=K\binom{t+m}{m+1},\qquad
 \bar c_t=|\mathcal C_t|=K\binom{T-t+m+1}{m+1},
\end{equation}
and hence
\begin{equation}\label{eq:totalsize}
 \sum_{t=1}^{T}(\bar r_t+\bar c_t)=2K\binom{T+m+1}{m+2}=\OmK(T^{m+2}).
\end{equation}

\begin{lemma}[Merged orders]\label{lem:orders}
All row and column orders, and their feasibility-filtered versions, can be constructed in $\OmK(T^{m+2})$ total time and space.
\end{lemma}
\begin{proof}
Write the keys in terms of invariant parts:
\begin{align}
 \beta_t(u,\gamma,\nvec)
 &=\underbrace{\gamma+w(\nvec)+Q_{u-1}}_{K^L(u,\gamma,\nvec)}-Q_t,\notag\\
 a_t(v,\eta,\Nvec)
 &=\underbrace{\eta+Q_v-w(\Nvec)}_{K^R(v,\eta,\Nvec)}-Q_t.\label{eq:invariantkeys}
\end{align}
The common translation does not affect a batch order. For rows, process $t=1,\ldots,T$. States from $\mathcal R_{t-1}$ persist with the same invariant keys. New states have $|\nvec|=t-u$ and total number $K\binom{t+m-1}{m}=\OmK(T^m)$. Sort only these states and merge them into the previous order. New-state sorting sums to $\OmK(T^{m+1}\log(T+2))$, while merging and storage cost $\OmK(T^{m+2})$ by~\eqref{eq:totalsize}. For fixed $m,K$, the latter bound absorbs the former.

For columns, process $t=T,\ldots,1$, adding states with $|\Nvec|=v-t$ and merging in nonincreasing $K^R$ order. The same count and bound apply. A scan per batch removes infinite table, stock-cost, and continuation values. Deletion preserves orders and interval domains. The total filtering cost is bounded by~\eqref{eq:totalsize}.
\end{proof}
Order reuse is an ingredient of the piecewise-linear method of \citet{papadopoulos2026}. Here the invariant universes include all boundary labels and support the nonlinear reduction. No integrality, distinct breakpoint sums, or presorted-input assumption is used.

\subsection{Backward evaluation, correctness, and complexity}\label{subsec:sweep}
Initialize $\mathrm{best}[u,\gamma]=Z_{u,\gamma}$ for each suffix start and boundary level. Process designated periods from $T$ down to $1$. In batch $t$, every column has continuation index $v+1\geq t+1$, so its required value at every boundary level has already been finalized. After updating all $u\leq t$, finalize $\Psi_{t,\gamma}=\mathrm{best}[t,\gamma]$ for every $\gamma\in\mathcal H$.

\begin{algorithm}[t]
\caption{Exact boundary-level Monge lot sizing}\label{alg:main}
\begin{singlespace}
\begin{algorithmic}[1]
\Require An instance satisfying Assumptions~\ref{ass:cost}--\ref{ass:stock}, terminal set $E$, initial stock $s_0\in\mathcal H$, and a batch-search routine
\Ensure Optimal value and witnesses, or $+\infty$ if infeasible
\State Compute all $L^\gamma,R^\eta$ and their choices by~\eqref{eq:left}--\eqref{eq:right}
\State Construct the combinatorial row and column orders of Lemma~\ref{lem:orders}
\State Set $\Psi_{T+1,\eta}\gets0$ for $\eta\in\mathcal H\cap E$, and $+\infty$ otherwise
\State Set $\mathrm{best}[u,\gamma]\gets Z_{u,\gamma}$ for all $u\leq T$ and $\gamma\in\mathcal H$
\For{$t=T,T-1,\ldots,1$}
 \State Filter the orders; form admitted rows and columns from~\eqref{eq:row}--\eqref{eq:col}
 \State Set $\mu_i\gets+\infty$ for every admitted row
 \For{each $J\in\mathcal J_t$}
  \State Obtain nonempty row intervals of $M_t^J$ by monotone scans
  \State Search their minima using Algorithm~\ref{alg:decompose} and staircase searching
  \State Update $\mu_i$ and its witness for each original row
 \EndFor
 \For{each admitted row $i=(u,\gamma,\nvec)$}
  \State $\mathrm{best}[u,\gamma]\gets\min\{\mathrm{best}[u,\gamma],\mu_i\}$; retain a witness
 \EndFor
 \State Set $\Psi_{t,\gamma}\gets\mathrm{best}[t,\gamma]$ for every $\gamma\in\mathcal H$
\EndFor
\State \Return $\Psi_{1,s_0}$ and its recursively stored witnesses
\end{algorithmic}
\end{singlespace}
\end{algorithm}
An empty row or column set contributes no updates; initialization and finalization still proceed. No infinite matrix is passed to a search primitive.

\begin{theorem}[Correctness]\label{thm:correct}
Algorithm~\ref{alg:main} returns the optimum under each of the terminal conventions in~\eqref{eq:terminalsets}, or $+\infty$ exactly when the instance is infeasible. Its witnesses recover an optimal production plan whenever the value is finite.
\end{theorem}
\begin{proof}
Use backward induction. At the start of iteration $t$, all $\Psi_{k,\eta}$ for $k\geq t+1$ and $\eta\in\mathcal H$ are final and correct. Previously processed batches had designated periods $q>t$ and used finalized continuations. For every $u\leq t$ and $\gamma\in\mathcal H$, the invariant is
\begin{equation}\label{eq:invariant}
 \mathrm{best}[u,\gamma]=\min\left\{Z_{u,\gamma},\ \min_{t<q\leq T}\Phi_{u,\gamma,q}\right\}.
\end{equation}
It holds initially with an empty inner set. All continuations required at $t$ are known. Proposition~\ref{prop:search} returns the exact row minima, inserting precisely $\Phi_{u,\gamma,t}$ into each incumbent. At $u=t$, the resulting value is $\Psi_{t,\gamma}$ by Proposition~\ref{prop:bellman}; a block beginning at $t$ cannot designate an earlier period. This proves finalization for every boundary label and the next invariant.

Every finite update has a feasible witness. Conversely, the structural lemma and Bellman equation retain a finite alternative for every feasible suffix. A witness identifies a terminal tail or $(u,\gamma,t,v,\eta,\nvec,\Nvec)$ with continuation at $(v+1,\eta)$. Recover the prefix and suffix from their table choices and set the designated production by~\eqref{eq:bridge}. Successive blocks are disjoint. Reconstruction takes $O(T)$ table-choice steps, or $O(mT)$ coordinate operations, and recovers all stock levels by balance.
\end{proof}

\begin{theorem}[Main complexity bounds]\label{thm:complexity}
Fix $m\geq1$ and $K\geq1$. Under Assumptions~\ref{ass:cost}--\ref{ass:stock} and the exact arithmetic/cost-oracle model, every terminal variant in~\eqref{eq:terminalsets} can be solved in
\begin{align}
 &\OmK\bigl(T^{m+2}\alpha(T+2)\bigr)&&\text{deterministic time},\notag\\
 &\OmK(T^{m+2})&&\text{expected Las Vegas time}.\label{eq:complexity}
\end{align}
Both algorithms have $\OmK(T^{m+2})$ worst-case space. The randomized algorithm is always correct; only its running time is random.
\end{theorem}
\begin{proof}
Tables and orders each take $\OmK(T^{m+2})$ time and space. There are $q=O(m)$ production domains per period. Sum Proposition~\ref{prop:search} over~\eqref{eq:totalsize}. A batch has $\OmK(T^{m+1})$ rows plus columns. Replacing the argument of $\alpha$ by a fixed-degree polynomial changes it by at most a constant depending on that degree and the fixed parameters, giving the deterministic bound. Conditional on previous calls, each randomized search has its expected linear bound. Summing the expectations gives $\OmK(T^{m+2})$. Tail initialization, incumbent updates, and reconstruction do not exceed preprocessing.

Store all tables and orders and process search instances sequentially. A single deterministic batch uses at most its operation count in workspace, namely $\OmK(T^{m+1}\alpha(T+2))=\OmK(T^{m+2})$. For a randomized search with expected cost at most $c_0(r+c)$, abort and restart any trial exceeding $2c_0(r+c)$ operations. Markov's inequality bounds the probability of abort by $1/2$. Independent trials retain expected linear time, use only $O(r+c)$ words, and return only correct completed outputs. Here $c_0$ is a search-routine constant, not the stock-boundary count. This proves worst-case space for the full algorithm.
\end{proof}
For a fixed batch, continuation offsets do not change during search. The searches are therefore offline; no online matrix data structure is required. The added boundary levels occur as row or column labels in~\eqref{eq:totalsize}, not as a separately enumerated sequence of stock regimes.

\subsection{An implementable rectangle-search variant}\label{subsec:rectangle}
The strongest bounds above use imported staircase primitives. Our implemented variant instead uses complete Monge matrices only. Build a balanced binary segment tree on the $c$ columns. Decompose each row interval into its $O(\log(c+2))$ disjoint canonical node intervals. At each node, collect assigned rows in their original order. The Cartesian product of these rows and the node's columns is completely defined and Monge. Search it by SMAWK, then minimize each row's answers over its assigned nodes.

\begin{algorithm}[t]
\caption{Canonical rectangles for one production domain}\label{alg:rectangle}
\begin{singlespace}
\begin{algorithmic}[1]
\Require Ordered matrix-entry oracle and nonempty intervals $[\lambda_i,\rho_i]$
\State Build a balanced segment tree over the columns
\For{each row $i$ in its original order}
 \State Append $i$ to each canonical node in the disjoint cover of $[\lambda_i,\rho_i]$
\EndFor
\State Initialize all row minima to $+\infty$
\For{each node with assigned rows}
 \State Search its complete row--column rectangle with SMAWK
 \State Update the minima and witnesses of its assigned rows
\EndFor
\State \Return the row minima and witnesses
\end{algorithmic}
\end{singlespace}
\end{algorithm}

\begin{lemma}[Rectangle-search bound]\label{lem:rectangle}
Algorithm~\ref{alg:rectangle} finds all row minima in $O((r+c)\log(c+2))$ deterministic time.
\end{lemma}
\begin{proof}
There are $O(r\log(c+2))$ row occurrences. At each tree depth, node column spans are disjoint and contain at most $c$ columns in total. Thus there are $O(c\log(c+2))$ column occurrences across all searched rectangles. Sum the linear SMAWK costs. Every defined entry is covered, and every searched entry is defined, proving correctness.
\end{proof}

\begin{corollary}[Implemented complexity]\label{cor:implemented}
Replacing the batch searches in Algorithm~\ref{alg:main} by Algorithm~\ref{alg:rectangle} gives a deterministic $\OmK(T^{m+2}\log(T+2))$ algorithm with $\OmK(T^{m+2})$ space.
\end{corollary}
\begin{proof}
Sum Lemma~\ref{lem:rectangle} over all periods and domains using~\eqref{eq:totalsize}. At one period, rectangle row-assignment lists occupy at most $\OmK(T^{m+1}\log(T+2))=\OmK(T^{m+2})$ space and can be discarded afterwards.
\end{proof}

\section{Computational experiments}\label{sec:experiments}
The experiments ask whether the implemented reduction returns the same objectives as the predecessor and whether the reduced transition work yields practical gains on the tested instances. They evaluate the rectangle/SMAWK variant of Corollary~\ref{cor:implemented}, denoted R-SMAWK, against a fresh implementation of the Koca--Yaman--Akt\"urk interval algorithm, denoted KYA. All reported timings use the archived experiment of 27 September 2026. No observations from the earlier manuscript experiment are combined with these data.

The timing experiment is restricted to $\mathcal H=\{0\}$, $s_0=0$, nonnegative inventory, and globally concave holding costs, with zero or free nonnegative terminal stock. Its 300 case timings do not measure the additional boundary-level implementation. Section~\ref{subsec:generalvalidation} reports the separate generalized-model correctness study.

\subsection{Original-model implementations and benchmark design}\label{subsec:impl}
KYA implements the original $F$, $G$, and $\widehat G$ recurrences and the shortest-path calculation \citep[Section~3]{koca2014}. The $F$ states contain $m$ past counts. For an exceptional production below $b_m$, $G$ uses $m$ past counts but only $m-1$ explicit future counts; the future count at $b_m$ is inferred using the original floor formula. For a quantity above $b_m$, $\widehat G$ uses $m$ future counts. This last table is omitted when all production is capped at $b_m$. The baseline therefore preserves the predecessor's smaller capacitated bound. It rolls layers and stores finite reachable states sparsely. It is not an explicit scan of the new algorithm's transition matrices.

The timed R-SMAWK code is the $K=1$ no-backlogging specialization of Corollary~\ref{cor:implemented}; it includes the prefix/suffix tables, merged orders, monotone domain scans, canonical rectangles, and backward sweep. Both programs evaluate the supplied nonlinear costs directly. They use sparse dictionaries and exact arbitrary-precision numbers; the theoretical bounds concern the array/cost-oracle model, not worst-case Python dictionary or bit-level runtimes. The KYA implementation requires strictly positive internal inventory within each interval, whereas R-SMAWK permits internal zeros; both represent the same horizon-level optimum. Generic zero-demand transitions charge actual zero-quantity costs. All timed instances have $p_t(0)=h_t(0)=0$.

The benchmark design is shown in Table~\ref{tab:design}. Each $(m,T)$ uses five base seeds, crossed with capped/uncapped production and zero/free terminal inventory. There are 75 base configurations, 150 serialized cost instances, and 300 model/convention cases. Related variants are not independent statistical samples.
\begin{table}[t]
\centering\singlespacing
\caption{Instance design. Each horizon uses five seeds and four capacity/terminal conventions.}\label{tab:design}
\begin{tabular}{cll}\toprule
$m$ & Positive breakpoints & Horizons $T$ \\\midrule
1 & 10 & 8, 16, 32, 48, 64 \\
2 & 6, 13 & 6, 10, 16, 22, 28 \\
3 & 4, 9, 15 & 4, 6, 10, 14, 18 \\\bottomrule
\end{tabular}
\end{table}
Demands are independent uniform draws from $\{0,\ldots,b_m\}$, so every timed case is feasible. Inside bounded piece $\ell$, the production formula is
\begin{equation}\label{eq:generator}
 f_{t\ell}(x)=a_{t\ell}x+z_{t\ell}+c_{t\ell}(b_\ell^2-x^2),
\end{equation}
with independent integer draws $a_{t\ell}\in\{0,\ldots,24\}$, $z_{t\ell}\in\{0,\ldots,39\}$, and $c_{t\ell}\in\{1,\ldots,4\}$. Every bounded interior is genuinely nonlinear and strictly concave. The available last ray is affine, with slope in $\{0,\ldots,19\}$ and intercept in $\{0,\ldots,49\}$. Set $p_t(0)=0$ and the cost at each positive shared breakpoint to the minimum of incident active formulas. Lowering an endpoint preserves concavity on each closed piece. All finite costs are nonnegative.

Holding costs are $h_t(s)=u_ts+v_t\min\{s,L_t\}$, with independent uniform integer coefficients $u_t\in\{0,\ldots,4\}$, $v_t\in\{0,\ldots,8\}$, and $L_t\in\{1,\ldots,7\}$. They are concave and can be nonaffine. The timed family has affine unbounded production tails; nonlinear tails and curved holding functions are examined separately in validation.

Capped and uncapped families share the base demands and formula coefficients. Removing the tail can also change the actual endpoint minimum at $b_m$. Their between-family comparison does not isolate the effect of a capacity constraint alone. Within each timed pair, however, both solvers receive exactly the same instance and terminal convention.

\subsection{Timing protocol and record provenance}\label{subsec:protocol}
The recorded experiment used single-threaded CPython 3.13.5 on an AMD EPYC 9V74 processor in a shared Linux container, with affinity pinned to one logical CPU. Each solver was warmed up once on a small instance. Case order was shuffled reproducibly; each case had three runs per solver, alternating the solver called first. Wall time was measured with \texttt{perf\_counter\_ns}; process CPU time was also recorded. The timer covered the complete objective-only solve, including preprocessing, orders, dynamic programming, searching, and in-call cleanup. Instance generation, explicit pre-call garbage collection, logging, and the correctness oracle were excluded. Ordinary garbage collection remained enabled during the call. No cache or state was reused across solves. Reconstruction was validated separately but was not timed.

The initial raw timing file contained NUL padding and lacked 108 observations for 18 complete cases; the cause was not established. Both solvers were rerun for all three repetitions on exactly these cases, selected solely by record incompleteness, using the same code, inputs, affinity, and protocol. The final dataset contains 1,692 retained original observations and 108 replacement observations, totaling 1,800 analyzed calls. The original damaged log, original case summaries, replacement observations, and affected-case list are preserved. Missing times were not inferred from summaries. Supplementary Material~S1 provides the provenance and a sensitivity summary excluding all repaired cases.

The source fingerprints, raw-objective consistency, complete call counts, and all timing aggregates were audited again during preparation of this version. The archived timings were not replaced in that audit. The accompanying \texttt{generate\_assets.py} derives the tables and figures directly from the raw records.

\subsection{Original-model exact-answer validation}\label{subsec:validation}
All 300 benchmark cases have identical KYA and R-SMAWK objectives, each confirmed by a separate inventory-state dynamic program. That formulation enumerates production and stock, without breakpoint counts, regeneration intervals, or Monge searching. By Corollary~\ref{cor:integer}, the oracle is exact on the integer benchmark inputs rather than a numerical grid approximation. Supplementary Material~S1 states its recurrence and proves its finite bounds.

A separate suite has 480 random integer cases, 60 cases with nonlinear concave unbounded production rays and exact rational costs, 60 exact rational quantity rescalings, and 40 targeted boundary cases. All 640 comparisons agree, including 132 infeasible cases. Every finite reconstructed R-SMAWK plan passes balance, domain, terminal-convention, and exact-cost checks. The suite also passes 500 complete-SMAWK tests, 1,000 partial-rectangle tests, and 60 merged-order tests. These checks were rerun successfully during preparation of the preceding publication draft; the records are preserved. The maximum finite objective difference is zero; no floating-point tolerance is used for objective agreement.

Boundary cases include mandatory terminal surplus, unavailable ranges, isolated quantities, colliding count weights, downward endpoint discontinuities, nonzero costs at zero, and curved holding functions. Rational rescalings are checked against the equivalent integer formulation using exact fractions. These results support implementation correctness but do not replace the mathematical proofs or test the unimplemented staircase routines. Table~\ref{tab:objectives} reports fixed replicate-zero examples, chosen independently of observed speed.
\begin{table}[t]
\centering\singlespacing
\caption{Exact objectives for replicate zero at the largest uncapacitated, free-terminal horizons.}\label{tab:objectives}
\begin{tabular}{rrrrrr}
\toprule
$m$ & $T$ & Seed & KYA & R-SMAWK & Inventory DP \\
\midrule
1 & 64 & 20367327 & 2376 & 2376 & 2376 \\
2 & 28 & 20463727 & 1657 & 1657 & 1657 \\
3 & 18 & 20562727 & 1367 & 1367 & 1367 \\
\bottomrule
\end{tabular}
\unskip
\end{table}

\subsection{Generalized-model implementation and validation}\label{subsec:generalvalidation}
A separate module, \texttt{extensions/extended\_lot\_sizing.py}, implements Algorithm~\ref{alg:main} with the canonical-rectangle searches of Algorithm~\ref{alg:rectangle} and plan reconstruction. It reuses frozen count-state and complete-matrix-search helpers without changing either timed solver or the archived performance data. Its boundary-indexed recurrences allow signed stock, several stock-cost regimes, and period-dependent domain availability.

All 624 policy-specific validation cases agree exactly with a finite inventory-state oracle. There are 39 infeasible cases; all 585 finite reconstructed plans pass balance, production-domain, stock-domain, terminal-policy, and exact-objective checks. The maximum finite objective gap is zero without a numerical tolerance. Selected cases also pass 192 full transition-scan comparisons, and all 96 original-model regression cases agree with both earlier structural solvers. The full scan is a diagnostic for the new matrices, not an implementation of Koca's generalized recurrence.

The suite covers ordinary backlogging, nonlinear piecewise-concave nonnegative holding, their combination, stock and backlog limits, minimum production quantities, exact rational rescalings, and targeted boundary/terminal examples. It contains 212 serialized cost models, with terminal-policy variants counted separately. The two-fractional-production example in Section~\ref{subsec:structure} is included. Supplementary Material~S1 gives all family counts and records and proves the lattice and finite-production bounds that make the inventory-state oracle exact, not a grid approximation.

These are correctness checks, not an extension timing benchmark or independent replication. In particular, the time ratios in Section~\ref{subsec:results} are not evidence of a measured speedup for backlogging or piecewise-concave stock costs. Neither the stronger staircase routines nor Koca's generalized recurrence were benchmarked on these cases.

\subsection{Timing results}\label{subsec:results}
For case $i$, let $\widetilde t_i^K$ and $\widetilde t_i^R$ be the three-run median wall times. We report
\begin{equation}\label{eq:ratios}
 r_i=\frac{\widetilde t_i^K}{\widetilde t_i^R},\qquad
 G(I)=\exp\left(\frac{1}{|I|}\sum_{i\in I}\log r_i\right).
\end{equation}
Values above one favor R-SMAWK. The geometric mean summarizes paired relative time; it is neither a ratio of total times nor a significance test.

R-SMAWK is faster on \NewWins\ cases and KYA on \OldWins. Table~\ref{tab:aggregate} gives $G=\GeoUncap$ for uncapacitated cases, $G=\GeoCap$ for capacitated cases, and $G=\GeoAll$ overall. The capacitated aggregate is effectively tied, despite differences between individual cases. This does not support uniform practical dominance.
\begin{table}[t]
\centering\singlespacing
\caption{Paired time comparison using three-run case medians. $G>1$ favors R-SMAWK.}\label{tab:aggregate}
\begin{tabular}{lrrrr}
\toprule
Family & Cases & R-SMAWK faster & KYA faster & $G$ \\
\midrule
Uncapacitated & 150 & 114 & 36 & 2.531 \\
Capacitated & 150 & 72 & 78 & 1.006 \\
All & 300 & 186 & 114 & 1.595 \\
\bottomrule
\end{tabular}
\unskip
\end{table}

Table~\ref{tab:times} reports every free-terminal size cell, the convention directly matching the predecessor's main model. Each displayed time is the median of the five instance medians. The reported cell ratio $R$ divides these two time aggregates; it is not $G$ and is not the median of the five paired ratios. The zero-terminal results and complete records are in Supplementary Material~S1.
\begin{table}[t]
\centering\singlespacing\small
\caption{Free-terminal timing results. Each time is the median of five three-run instance medians; $R$ is the ratio of these time aggregates.}\label{tab:times}
\setlength{\tabcolsep}{4pt}
\begin{tabular}{rr rrr rrr}
\toprule
& & \multicolumn{3}{c}{Uncapacitated} & \multicolumn{3}{c}{Capacitated} \\
\cmidrule(lr){3-5}\cmidrule(lr){6-8}
$m$ & $T$ & KYA (s) & R-SMAWK (s) & $R$ & KYA (s) & R-SMAWK (s) & $R$ \\
\midrule
1 & 8 & 0.000612 & 0.000749 & 0.82 & 0.000407 & 0.000647 & 0.63 \\
1 & 16 & 0.009243 & 0.004339 & 2.13 & 0.002605 & 0.003068 & 0.85 \\
1 & 32 & 0.141581 & 0.039559 & 3.58 & 0.037342 & 0.024966 & 1.50 \\
1 & 48 & 1.096887 & 0.155102 & 7.07 & 0.154150 & 0.090649 & 1.70 \\
1 & 64 & 3.947976 & 0.366945 & 10.76 & 0.464267 & 0.228154 & 2.03 \\
\addlinespace
2 & 6 & 0.000680 & 0.000832 & 0.82 & 0.000533 & 0.000796 & 0.67 \\
2 & 10 & 0.005416 & 0.004426 & 1.22 & 0.002558 & 0.003594 & 0.71 \\
2 & 16 & 0.071398 & 0.031275 & 2.28 & 0.025666 & 0.024344 & 1.05 \\
2 & 22 & 0.471173 & 0.111517 & 4.23 & 0.130587 & 0.089803 & 1.45 \\
2 & 28 & 2.333821 & 0.285644 & 8.17 & 0.368388 & 0.218658 & 1.68 \\
\addlinespace
3 & 4 & 0.000389 & 0.000563 & 0.69 & 0.000334 & 0.000590 & 0.57 \\
3 & 6 & 0.001411 & 0.001906 & 0.74 & 0.001078 & 0.001826 & 0.59 \\
3 & 10 & 0.030599 & 0.014822 & 2.06 & 0.006475 & 0.011014 & 0.59 \\
3 & 14 & 0.257344 & 0.069962 & 3.68 & 0.057691 & 0.057457 & 1.00 \\
3 & 18 & 1.941863 & 0.238541 & 8.14 & 0.248634 & 0.174139 & 1.43 \\
\bottomrule
\end{tabular}
\unskip
\end{table}
At the largest uncapacitated free-terminal horizons, the cell ratios are 10.76 for $(m,T)=(1,64)$, 8.17 for $(2,28)$, and 8.14 for $(3,18)$. The corresponding capacitated ratios are 2.03, 1.68, and 1.43. The smallest tested free-terminal groups favor KYA in all six series. Figures~\ref{fig:uncap} and~\ref{fig:cap} show how these ratios vary across the tested horizons. The connecting lines are not fitted complexity curves.

\begin{figure}[t]
\centering
\includegraphics[width=0.82\textwidth]{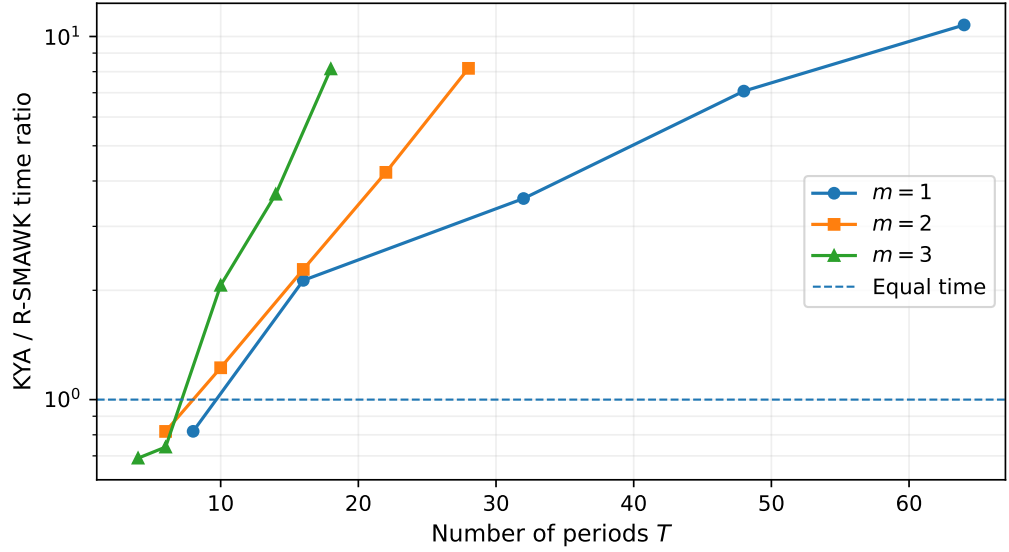}
\caption{Uncapacitated, free-terminal instances. Points show the ratio $R$ of cell median times from Table~\ref{tab:times}. Each summarizes five cases and three runs per solver; values above one favor R-SMAWK.}\label{fig:uncap}
\end{figure}
\begin{figure}[t]
\centering
\includegraphics[width=0.82\textwidth]{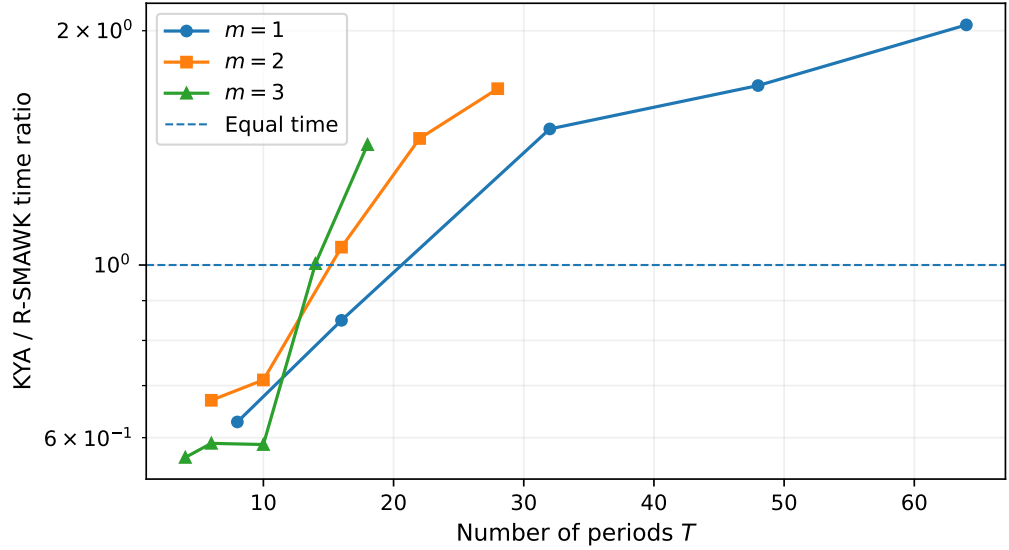}
\caption{Capacitated, free-terminal instances. Aggregation is the same as in Figure~\ref{fig:uncap}. Larger tested cells favor R-SMAWK, whereas the overall capacitated geometric-mean ratio is close to one.}\label{fig:cap}
\end{figure}

\subsection{Interpretation, limitations, and research workflow}\label{subsec:limitations}
The observations are consistent with a trade-off between matrix-search setup work and avoided state combinations. They do not identify the causal contribution of any one optimization: no ablation study was performed. Nor do these finite horizons establish the asymptotic exponents. The study uses $m\leq3$, one timed breakpoint grid per $m$, one synthetic family, five seeds per cell, and three repetitions. The small cases and near-ties are particularly sensitive to interpreter and shared-container effects.

Peak memory, schedule-reconstruction time, industrial instances, commercial solvers, compiled implementations, and the stronger staircase routines were not benchmarked. The inventory-state dynamic program is a correctness oracle, not a timed competitor; its favorable small integer ranges do not support claims about superiority over pseudo-polynomial methods. No general crossover horizon or industrial performance claim follows from these data.

ChatGPT (OpenAI; GPT-6 Astra Pro) was used in developing the reduction and proof drafts, writing the solver and validation code, scripting experiments, and preparing the manuscript and reproducible Matplotlib visualizations. Objective values and elapsed times were obtained by executing the archived programs, not by language-model prediction. The inventory oracle is independent in formulation, but all implementations were developed in the same AI-assisted workflow; this is not an independent replication by another research group.

\FloatBarrier
\section{Special cases and scope}\label{sec:special}
The model in Section~\ref{sec:model} and the single algorithm in Section~\ref{sec:algorithm} include the following cases. These are specializations of the main result, not separate algorithms requiring duplicate proofs.

\subsection{Concave holding costs without backlogging}
Set $s_0=0$, $\mathcal H=\{0\}$, and allow only the ray $[0,+\infty)$ with a concave holding cost. Boundary intervals are the usual zero-inventory regeneration intervals. Theorem~\ref{thm:complexity} reduces to $\Om(T^{m+2}\alpha(T+2))$ deterministic time or $\Om(T^{m+2})$ expected Las Vegas time, and Corollary~\ref{cor:implemented} gives $\Om(T^{m+2}\log(T+2))$ for rectangle searching. Both zero and free nonnegative terminal stock are represented. This is the shared problem underlying the comparison with the main algorithm of \citet{koca2014} and the timed instances in Section~\ref{sec:experiments}.

\subsection{Backlogging}
For ordinary backlogging, let $\mathcal H=\{0\}$ and allow both half-lines, with costs concave separately on each. Then $K=1$ and the same $\Om$ bounds hold. Internal stock between zero-boundary times need only be nonzero; its sign may change between successive periods without a zero stock value. The perturbation proof already uses the appropriate concave piece at each affected period.

The terminal set must match the intended service requirement. Taking $E=\{0\}$ clears inventory and backlog, whereas $E=[0,+\infty)$ serves all demand but permits surplus. Taking $E=\mathbb R$ intentionally allows a final backlog with its specified cost. The latter is not a substitute for a demand-clearing terminal convention.

\subsection{Piecewise-concave stock costs and stock limits}
Include all common finite holding--backlog breakpoints in $\mathcal H$, together with zero, the initial stock, and any finite stock-domain limits. Assumption~\ref{ass:stock} permits positive and negative stock regimes simultaneously. Inventory or backlog limits at these levels are enforced by removing unavailable stock domains. Their availability may depend on the period. For fixed $K$, these cases preserve the exponent of $T$ in the main bounds. Thus the additional holding-cost and backlogging cases described in Section~3.4 of \citet{koca2014} fit the same matrix formulation under the stated assumptions.

\subsection{Production capacities, minimum quantities, and regimes}
A production capacity at a common breakpoint removes all domains above it. A minimum order quantity $L$ and capacity $C$ allow the production domain $\{0\}\cup[L,C]$. The usual two-regime regular-production/subcontracting case has one finite production breakpoint. Table~\ref{tab:extcases} specializes the bounds for common production thresholds, which must belong to the supplied production grid. Its predecessor column refers to that paper's stated dynamic program, not to a claim that these are the best bounds currently known for each narrower special case.

\begin{table}[ht]
\centering\small
\caption{Specializations with fixed stock-boundary count. The rectangle implementation has an additional logarithmic factor relative to the expected bound.}\label{tab:extcases}
\begin{tabular}{@{}lccc@{}}
\toprule
Model & Koca et al. DP & Deterministic staircase & Expected Las Vegas\\
\midrule
Common production capacity ($m=1$) & $O(T^4)$ & $O(T^3\alpha(T+2))$ & $O(T^3)$\\
Minimum quantity and capacity ($m=2$) & $O(T^6)$ & $O(T^4\alpha(T+2))$ & $O(T^4)$\\
Two production regimes ($m=1$) & $O(T^5)$ & $O(T^3\alpha(T+2))$ & $O(T^3)$\\
\bottomrule
\end{tabular}
\end{table}

\subsection{Affine production pieces}
If $p_t(x)=\kappa+\sigma x$ on a search interval, then
\begin{equation}\label{eq:affine}
 M_t(i,j)=(A_i-\sigma\beta_i+\kappa)+(C_j+\sigma a_j).
\end{equation}
The column score is independent of the row, including when columns have different ending boundary levels. Processing rows in reverse order makes both window endpoints nondecreasing, so a monotone deque maintains the minimum score with each column entering and leaving at most once. If an endpoint's actual cost differs from an adjoining affine formula, use separate singleton endpoint searches and open affine interiors with the appropriate strict comparisons. The closed-domain structural assumption is unchanged. This specialization is not a new priority claim for piecewise-linear lot sizing; it is consistent with the related bounds of \citet{ou2017,papadopoulos2026}.

\subsection{Limits of the result}
The theorem requires a fixed number of common production and stock boundary levels, actual concavity on the supplied closed domains, consistent endpoint values, and exact cost queries. Arbitrary holding functions, arbitrary half-open endpoint conventions, multiple items with shared capacities, and lost-sales decisions are not covered. Period-specific production thresholds or stock limits may be represented by their common unions, but the fixed-$m,K$ bounds do not give a uniform guarantee independent of the size of those unions. The expected bound matches this algorithm's explicit state order, not a lower bound for the optimization problem.

\section{Conclusions}\label{sec:conclusion}
Fixing a designated production period exposes a piecewise-Monge transition matrix for lot sizing with nonlinear piecewise-concave production costs. Searching these matrices avoids explicit combinations of left and right count states. The unified boundary-labelled prefix and suffix tables include backlogging and piecewise-concave holding costs, with a fixed number $K$ of common stock levels. Reused orders and backward evaluation give $O_{m,K}(T^{m+2}\alpha(T+2))$ deterministic and $O_{m,K}(T^{m+2})$ expected Las Vegas bounds. The result covers the predecessor's additional model cases under explicit production, stock-domain, and terminal assumptions.

The implemented rectangle/SMAWK variant has one additional logarithmic factor. All 300 no-backlogging benchmark objectives agree with the predecessor implementation and an inventory-state oracle; the original 640-case validation suite also passes. A separate 624-case suite confirms exact objective agreement and feasible-plan reconstruction for the extended implementation, without supplying extension timing evidence. Larger uncapacitated tests show substantial gains, while many smaller and capacitated cases favor the predecessor. The contribution is therefore an exact reduction and a documented implementation comparison, not a claim of universal practical superiority. Implementing the stronger primitives, measuring memory and reconstruction, obtaining independent replication, and testing broader cost families are the next empirical steps. A deterministic $\OmK(T^{m+2})$ nonlinear-batch method would require further algorithmic work.

\section*{Data and code availability}
The accompanying reproducibility archive contains both timed solvers, the inventory-state oracle, serialized instances, all raw timing records, validation programs, source fingerprints, and the original and replacement logs. Supplementary Material~S1 documents reproduction and the oracle bounds. The tables and figures are regenerated by \texttt{generate\_assets.py}. The separate \texttt{extensions/} directory contains the boundary-level solver, its exact validation driver, all 624 case records and serialized inputs, and source hashes. The archive does not implement the Klawe--Kleitman or Chan staircase algorithms.

\section*{Declaration of competing interest}
The author declares no competing interests.
\section*{Declaration of generative AI and AI-assisted technologies in the manuscript preparation process}
ChatGPT (OpenAI; GPT-6 Astra Pro) was used to develop and edit mathematical exposition, literature positioning, code, LaTeX source, and the presentation of computational results. Its research role is described in Section~\ref{subsec:limitations}.

\end{document}